\documentclass[pra,aps,twocolumn,showpacs,twoside,superscriptaddress]{revtex4}

\usepackage{amsmath,amsfonts,amssymb,caption,color,epsfig,graphics,graphicx,hyperref,latexsym,mathrsfs,revsymb,theorem,url,verbatim,enumerate,epstopdf,tikz,float}
\hypersetup{colorlinks,linkcolor={blue},citecolor={red},urlcolor={blue}}
\usepackage{mathtools}

\newtheorem{definition}{Definition}
\newtheorem{proposition}[definition]{Proposition}
\newtheorem{lemma}[definition]{Lemma}

\newtheorem{theorem}[definition]{Theorem}
\newtheorem{corollary}[definition]{Corollary}
\newtheorem{conjecture}[definition]{Conjecture}

\newtheorem{remark}[definition]{Remark}
\newtheorem{example}[definition]{Example}
\newtheorem{question}[definition]{Question}

\def\bcj{\begin{conjecture}}
\def\ecj{\end{conjecture}}
\def\bcr{\begin{corollary}}
\def\ecr{\end{corollary}}
\def\bd{\begin{definition}}
\def\ed{\end{definition}}
\def\bea{\begin{eqnarray}}
\def\eea{\end{eqnarray}}
\def\bem{\begin{enumerate}}
\def\eem{\end{enumerate}}
\def\bex{\begin{example}}
\def\eex{\end{example}}
\def\bim{\begin{itemize}}
\def\eim{\end{itemize}}
\def\bl{\begin{lemma}}
\def\el{\end{lemma}}
\def\bma{\begin{bmatrix}}
\def\ema{\end{bmatrix}}
\def\bpf{\begin{proof}}
\def\epf{\end{proof}}
\def\bpp{\begin{proposition}}
\def\epp{\end{proposition}}
\def\bqu{\begin{question}}
\def\equ{\end{question}}
\def\br{\begin{remark}}
\def\er{\end{remark}}
\def\bt{\begin{theorem}}
\def\et{\end{theorem}}

\def\squareforqed{\hbox{\rlap{$\sqcap$}$\sqcup$}}
\def\qed{\ifmmode\squareforqed\else{\unskip\nobreak\hfil
\penalty50\hskip1em\null\nobreak\hfil\squareforqed
\parfillskip=0pt\finalhyphendemerits=0\endgraf}\fi}
\def\endenv{\ifmmode\;\else{\unskip\nobreak\hfil
\penalty50\hskip1em\null\nobreak\hfil\;
\parfillskip=0pt\finalhyphendemerits=0\endgraf}\fi}
\newenvironment{proof}{\noindent \textbf{{Proof.~} }}{\qed}
\def\Dbar{\leavevmode\lower.6ex\hbox to 0pt
{\hskip-.23ex\accent"16\hss}D}
\makeatletter
\def\url@leostyle{%
  \@ifundefined{selectfont}{\def\UrlFont{\sf}}{\def\UrlFont{\small\ttfamily}}}
\makeatother
\def\bcj{\begin{conjecture}}
\def\ecj{\end{conjecture}}
\def\bcr{\begin{corollary}}
\def\ecr{\end{corollary}}
\def\bd{\begin{definition}}
\def\ed{\end{definition}}
\def\bea{\begin{eqnarray}}
\def\eea{\end{eqnarray}}
\def\bem{\begin{enumerate}}
\def\eem{\end{enumerate}}
\def\bex{\begin{example}}
\def\eex{\end{example}}
\def\bim{\begin{itemize}}
\def\eim{\end{itemize}}
\def\bl{\begin{lemma}}
\def\el{\end{lemma}}
\def\bpf{\begin{proof}}
\def\epf{\end{proof}}
\def\bpp{\begin{proposition}}
\def\epp{\end{proposition}}
\def\bqu{\begin{question}}
\def\equ{\end{question}}
\def\br{\begin{remark}}
\def\er{\end{remark}}
\def\bt{\begin{theorem}}
\def\et{\end{theorem}}

\def\btb{\begin{tabular}}
\def\etb{\end{tabular}}

\newcommand{\nc}{\newcommand}

\def\x{\xi}

 \nc{\bbA}{\mathbb{A}} \nc{\bbB}{\mathbb{B}} \nc{\bbC}{\mathbb{C}}
 \nc{\bbD}{\mathbb{D}} \nc{\bbE}{\mathbb{E}} \nc{\bbF}{\mathbb{F}}
 \nc{\bbG}{\mathbb{G}} \nc{\bbH}{\mathbb{H}} \nc{\bbI}{\mathbb{I}}
 \nc{\bbJ}{\mathbb{J}} \nc{\bbK}{\mathbb{K}} \nc{\bbL}{\mathbb{L}}
 \nc{\bbM}{\mathbb{M}} \nc{\bbN}{\mathbb{N}} \nc{\bbO}{\mathbb{O}}
 \nc{\bbP}{\mathbb{P}} \nc{\bbQ}{\mathbb{Q}} \nc{\bbR}{\mathbb{R}}
 \nc{\bbS}{\mathbb{S}} \nc{\bbT}{\mathbb{T}} \nc{\bbU}{\mathbb{U}}
 \nc{\bbV}{\mathbb{V}} \nc{\bbW}{\mathbb{W}} \nc{\bbX}{\mathbb{X}}
 \nc{\bbZ}{\mathbb{Z}}

 \nc{\bA}{{\bf A}} \nc{\bB}{{\bf B}} \nc{\bC}{{\bf C}}
 \nc{\bD}{{\bf D}} \nc{\bE}{{\bf E}} \nc{\bF}{{\bf F}}
 \nc{\bG}{{\bf G}} \nc{\bH}{{\bf H}} \nc{\bI}{{\bf I}}
 \nc{\bJ}{{\bf J}} \nc{\bK}{{\bf K}} \nc{\bL}{{\bf L}}
 \nc{\bM}{{\bf M}} \nc{\bN}{{\bf N}} \nc{\bO}{{\bf O}}
 \nc{\bP}{{\bf P}} \nc{\bQ}{{\bf Q}} \nc{\bR}{{\bf R}}
 \nc{\bS}{{\bf S}} \nc{\bT}{{\bf T}} \nc{\bU}{{\bf U}}
 \nc{\bV}{{\bf V}} \nc{\bW}{{\bf W}} \nc{\bX}{{\bf X}}
 \nc{\bZ}{{\bf Z}}

\nc{\cA}{{\cal A}} \nc{\cB}{{\cal B}} \nc{\cC}{{\cal C}}
\nc{\cD}{{\cal D}} \nc{\cE}{{\cal E}} \nc{\cF}{{\cal F}}
\nc{\cG}{{\cal G}} \nc{\cH}{{\cal H}} \nc{\cI}{{\cal I}}
\nc{\cJ}{{\cal J}} \nc{\cK}{{\cal K}} \nc{\cL}{{\cal L}}
\nc{\cM}{{\cal M}} \nc{\cN}{{\cal N}} \nc{\cO}{{\cal O}}
\nc{\cP}{{\cal P}} \nc{\cQ}{{\cal Q}} \nc{\cR}{{\cal R}}
\nc{\cS}{{\cal S}} \nc{\cT}{{\cal T}} \nc{\cU}{{\cal U}}
\nc{\cV}{{\cal V}} \nc{\cW}{{\cal W}} \nc{\cX}{{\cal X}}
\nc{\cZ}{{\cal Z}}

\nc{\hA}{{\hat{A}}} \nc{\hB}{{\hat{B}}} \nc{\hC}{{\hat{C}}}
\nc{\hD}{{\hat{D}}} \nc{\hE}{{\hat{E}}} \nc{\hF}{{\hat{F}}}
\nc{\hG}{{\hat{G}}} \nc{\hH}{{\hat{H}}} \nc{\hI}{{\hat{I}}}
\nc{\hJ}{{\hat{J}}} \nc{\hK}{{\hat{K}}} \nc{\hL}{{\hat{L}}}
\nc{\hM}{{\hat{M}}} \nc{\hN}{{\hat{N}}} \nc{\hO}{{\hat{O}}}
\nc{\hP}{{\hat{P}}} \nc{\hR}{{\hat{R}}} \nc{\hS}{{\hat{S}}}
\nc{\hT}{{\hat{T}}} \nc{\hU}{{\hat{U}}} \nc{\hV}{{\hat{V}}}
\nc{\hW}{{\hat{W}}} \nc{\hX}{{\hat{X}}} \nc{\hZ}{{\hat{Z}}}

\nc{\hn}{{\hat{n}}}

\def\rank{\mathop{\rm rank}}

\def\tr{\mathop{\rm Tr}}

\newcommand{\bra}[1]{\langle#1|}
\newcommand{\ket}[1]{|#1\rangle}

\newcommand{\ketbra}[2]{|#1\rangle\!\langle#2|}
\newcommand{\braket}[2]{\langle#1|#2\rangle}

\newcommand{\fl}[2]{\lfloor\frac{#1}{#2}\rfloor}

\def\Dbar{\leavevmode\lower.6ex\hbox to 0pt
{\hskip-.23ex\accent"16\hss}D}

\begin{document}

\title{Unextendible product bases from tile structures  and their local entanglement-assisted distinguishability}


\pacs{03.65.Ud, 03.67.Mn}

\author{Fei Shi}
\email{shifei@mail.ustc.edu.cn}
\affiliation{School of Cyber Security,
	University of Science and Technology of China, Hefei, 230026, People's Republic of China}

\author{Xiande Zhang}
\email{drzhangx@ustc.edu.cn}
\affiliation{School of Mathematical Sciences,
	University of Science and Technology of China, Hefei, 230026, People's Republic of China}

\author{Lin Chen}
\email{linchen@buaa.edu.cn}
\affiliation{School of Mathematical Sciences, Beihang University, Beijing 100191, China}
\affiliation{International Research Institute for Multidisciplinary Science, Beihang University, Beijing 100191, China}

\begin{abstract}
We completely characterize the condition when a tile structure provides an unextendible product basis (UPB), and construct UPBs of different large sizes in $\bbC^m\otimes\bbC^n$ for any $n\geq m\geq 3$. This solves an open problem in [S. Halder \emph{et al.}, \href{https://journals.aps.org/pra/abstract/10.1103/PhysRevA.99.062329}{Phys. Rev. A \textbf{99}, 062329 (2019)}]. As an application, we show that our UPBs of size $(mn-4\fl{m-1}{2})$ in $\bbC^m\otimes\bbC^n$ can be  perfectly distinguished by local operations and classical communications assisted with a $\lceil\frac{m}{2}\rceil\otimes\lceil\frac{m}{2}\rceil$ maximally entangled state.


\end{abstract}

\maketitle



\section{Introduction}
\label{sec:int}

Unextendible product basis (UPB) is
 a set of orthonormal product states whose complementary space has no product states. They give a systematic construction of positive-partial-transpose (PPT) entangled states as follows  \cite{bennett1999unextendible}. Given a UPB $\{\ket{\psi_i}\}_{i=1}^t$ in $\bbC^m\otimes\bbC^n$, then the state
$
\rho=\frac{1}{mn-t}(I-\sum_{i=1}^{t}\ketbra{\psi_i}{\psi_i})
$
is a PPT entangled state.
UPBs are also connected to the quantum nonlocality without entanglement, Bell inequalities without quantum violation and fermionic system \cite{bennett1999unextendible,dms03, Tura2012Four, Chen2014Unextendible, Augusiak2012tight,augusiak2011bell}. In spite of much efforts devoted to the construction of UPBs of small size \cite{AL01,Fen06,Chen2013The}, there has been little progress on the construction of UPBs of large size. We shall address this problem, and it is the first motivation of this work.

Although  UPBs cannot be distinguished perfectly by local operations and classical communications (LOCC) \cite{de2004distinguishability}, Ref. \cite{cohen2008understanding} has shown the local distinguishability of  UPBs using LOCC protocols assisted by entanglement as a nonlocal resource. Further the UPB called  GenTiles2 in $\bbC^m\otimes\bbC^n$ ($m\leq n$) can be distinguished by LOCC with a $\lceil\frac{m}{2}\rceil\otimes\lceil\frac{m}{2}\rceil$ maximally entangled state \cite{cohen2008understanding,dms03}. Then, local distinguishability with entanglement as a resource
attracted more and more attention \cite{ghosh2001distinguishability,bandyopadhyay2016entanglement,zhang2016entanglement,gungor2016entanglement,zhang2018local}. Recently, it has been shown that some UPBs constructed from tile structures in $\bbC^m\otimes\bbC^m$ can be distinguished by LOCC with a $\lceil\frac{m}{2}\rceil\otimes\lceil\frac{m}{2}\rceil$ maximally entangled state when $m\geq 3$ is odd \cite{zhang2020locally,halder2019family}. In particular, Ref. \cite{halder2019family} wonders whether the construction of UPBs can be generalized to even-dimensional systems. Further, Ref. \cite{cohen2008understanding} asks whether other types of UPBs can be locally distinguished
by efficiently using entanglement resource. We shall give positive answers to both problems above. This is the second motivation of this work.

In this paper, we construct UPBs of large size, by constructing  tile structures illustrated in Figure \ref{Figure:example}. We begin by reviewing the connection of  UPBs and tile structures,  and introduce the U-tile structures in Definition \ref{df:utile}. We present the main result of this paper in Theorem \ref{Thm:U-tile}, that is, a tile structure with $s$-tiles corresponds to a UPB of size $(mn-s+1)$ in $\bbC^m\otimes \bbC^n$ if and only if this tile structure is a U-tile structure. By applying Theorem \ref{Thm:U-tile}, we generalize the construction in \cite{halder2019family} and show that	there exists a UPB of size $(mn-4\fl{m-1}{2})$ in $\bbC^m\otimes\bbC^n$ for $3\leq m\leq n$ in Proposition~\ref{pro:generalization}.
In Proposition~\ref{pro:upb42m-1}, we show that there is a UPB of size $(mn-k)$  in $\bbC^m\otimes\bbC^n$ for $4\leq m\leq n$, where $ 4\leq k\leq 2m-1$, and the maximum size of UPBs in $\bbC^m\otimes\bbC^n$ is $mn-4$ for $3\leq m\leq n$.
Finally, we show the UPB constructed from Proposition~\ref{pro:generalization} can be perfectly distinguished by LOCC with a $\lceil\frac{m}{2}\rceil\otimes\lceil\frac{m}{2}\rceil$ maximally entangled state in Theorem~\ref{Thm:mevendis}.

\begin{figure}[H]
	\centering
	\begin{tikzpicture}
	[
	box/.style={rectangle,draw=black,thick, minimum size=1cm},
	]
	
	\foreach \x in {0,1,...,3}{
		\foreach \y in {0,1,...,3}
		\node[box] at (\x,\y){};
	}
	\node[box] at (0,3){1};
	\node[box] at (1,3){1};
	\node[box]  at (2,3){2};
	\node[box]  at (2,0){2};
	\node[box]   at (3,3){3};
	\node[box]   at (3,2){3};
	\node[box]   at (3,1){3};
	\node[box]   at (1,2){4};
	\node[box]   at (1,1){4};
	\node[box]   at (1,0){4};
	\node[box]   at (0,0){5};
	\node[box]   at (3,0){5};
	\node[box]   at (0,2){6};
	\node[box]   at (0,1){6};
	\node[box]   at (2,2){6};
	\node[box]   at (2,1){6};
	\draw (-1,3) node {0};
	\draw (-1,2) node {1};
	\draw (-1,1) node {2};
	\draw (-1,0) node {3};
	\draw (-1.8,1.5) node {A};
	\draw (0,4) node {0};
	\draw (1,4) node {1};
	\draw (2,4) node {2};
	\draw (3,4) node {3};
	\draw (1.5,4.8) node {B};
	\end{tikzpicture}
	\caption{This is a tile structure of system $A,B$ in $\bbC^4\otimes \bbC^4$. A tile structure is a rectangle paved by some disjoint tiles. It gives a complete orthogonal product basis (COPB). We can obtain a UPB by  deleting some  states and adding a special state of COPB. We shall explain more details in Example \ref{ex:4x4upb=size11}.}  \label{Figure:example}
\end{figure}

We briefly review the task of distinguishing bipartite states by LOCC. Alice and Bob share a set of bipartite orthogonal states, and they don't know which state their system is in. Their aim is to determine the state by LOCC. It is shown that any two orthogonal pure states can be distinguished by LOCC \cite{walgate2000local}. There exists a product bases in $\bbC^3\otimes\bbC^3$ that cannot be  distinguished by LOCC \cite{bennett1999quantum}. Any three of Bell states cannot be  distinguished by LOCC \cite{ghosh2001distinguishability}. Our results on the construction of UPBs and their discrimination can be applied to these topics and produce more efficient protocols.

The rest of this paper is organized as follows. In Sec. \ref{sec:pre}, we introduce the preliminary knowledge used in this paper, such as UPBs and tile structures. In Sec. \ref{sec:utile} we connect U-tile structures and UPBs, and present the main result of this paper. In Sec. \ref{sec:app} we apply our results to investigate local distinguishability of UPBs by using entanglement resource. We conclude in Sec. \ref{sec:con}.

\section{Preliminary}
\label{sec:pre}
In this section we introduce the preliminary knowledge and facts.
Throughout this paper, we do not normalize states and operators for simplicity. Every bipartite pure state can be written as $\ket{\psi}=\sum_{i,j}m_{i,j}\ket{i}\ket{j}\in\bbC^m\otimes\bbC^n$, where $\ket{i}$ and $\ket{j}$ are the computational bases of $\bbC^m$ and $\bbC^n$, respectively. There exists a one to one correspondence between the state $\ket{\psi}$  and the $m\times n$ matrix $M=(m_{ij})$. If $\rank(M)=1$, then  $\ket{\psi}$ is a product state, and if $\rank(M)>1$ then $\ket{\psi}$ is an entangled state. For example, the state $\ket{00}+\ket{11}$ in $\bbC^{2}\otimes\bbC^{2}$ corresponds to the matrix
$M=\begin{pmatrix}
1 &0\\
0 &1
\end{pmatrix}$.
It is an entangled state since $\rank(M)=2$. Assume $\ket{\psi_i}$ corresponds to a matrix $M_i$, $i=1,2$, then $\braket{\psi_1}{\psi_2}=\tr(M_1^{\dagger}M_2)$, where $\braket{\psi_1}{\psi_2}$ is the inner product of $\ket{\psi_1}$ and $\ket{\psi_2}$.

To present the definition of UPBs, we consider the complete orthogonal product basis (COPB). This is a set of orthogonal product states that spans $\cH=\bbC^m\otimes\bbC^n$. The incomplete orthogonal product basis (ICOPB) is a set of pure orthogonal product states that spans a subspace $\cH_S$ of	$\cH$. An unextendible product basis (UPB) is an ICOPB such that there is no product state in $\cH_S^{\bot}$.

Now we define the tile structure in $\bbC^m\otimes \bbC^n$. This is an $m\times n$ rectangle $\cT$ paved by disjoint tiles $\{t_i\}$, denoted by $\cT=\cup_{i}t_i$. A tile $t_i$ must be a rectangle. In our notation, a rectangle could be  separated, that is, a set of cells that can be changed to a rectangle through row and column permutations.  In Figure~\ref{Figure:example}, it is a $4\times 4$ rectangle $\cT$ paved by $6$ disjoint tiles, where grids of the same index form a tile.  We have $\cT=\cup_{i=1}^6 t_i$. Denote $w_k=e^{\frac{2\pi\sqrt{-1}}{k}}$. Next we show how to construct a UPB of size $11$ by Figure \ref{Figure:example}.

\begin{example}
\label{ex:4x4upb=size11}
In Figure~\ref{Figure:example},
	tile $1$ gives two orthogonal states in $\bbC^4\otimes\bbC^4$, namely $\ket{0}(\ket{0}+\ket{1})$ and $\ket{0}(\ket{0}-\ket{1})$. Tile $2$ gives two orthogonal states $(\ket{0}+\ket{3})\ket{2}$ and $(\ket{0}-\ket{3})\ket{2}$. One can similarly derive the states for other tiles. Since tiles $i$ and $j$ are disjoint, we know that any state from tile $i$ is orthogonal to any state from tile $j$ for $1\leq i\neq j\leq 6$. As a result,
 Figure~\ref{Figure:example} provides a COPB as follows. Denote this basis by $\mathcal{B}$.
	\begin{align*}
	&\ket{\psi_1^{(1)}}=\ket{0}(\ket{0}+\ket{1}), &\ket{\psi_1^{(2)}}=\ket{0}(\ket{0}-\ket{1}),\\
	&\ket{\psi_2^{(1)}}=(\ket{0}+\ket{3})\ket{2}, &\ket{\psi_2^{(2)}}=(\ket{0}-\ket{3})\ket{2},\\
	&\ket{\psi_3^{(1)}}=(\ket{0}+\ket{1}+\ket{2})\ket{3},&\\	
	&\ket{\psi_3^{(2)}}=(\ket{0}+w_3\ket{1}+w_3^2\ket{2})\ket{3},&\\ &\ket{\psi_3^{(3)}}=(\ket{0}+w_3^2\ket{1}+w_3\ket{2})\ket{3},&\\
    &\ket{\psi_4^{(1)}}=(\ket{1}+\ket{2}+\ket{3})\ket{1},&\\	
    &\ket{\psi_4^{(2)}}=(\ket{1}+w_3\ket{2}+w_3^2\ket{3})\ket{1},&\\ &\ket{\psi_4^{(3)}}=(\ket{1}+w_3^2\ket{2}+w_3\ket{3})\ket{1},& \\
	&\ket{\psi_5^{(1)}}=\ket{3}(\ket{0}+\ket{3}),  &\ket{\psi_5^{(2)}}=\ket{3}(\ket{0}-\ket{3}), \\
	&\ket{\psi_6^{(1)}}=(\ket{1}+\ket{2})(\ket{0}+\ket{2}),&\\
	&\ket{\psi_6^{(2)}}=(\ket{1}+\ket{2})(\ket{0}-\ket{2}),&\\ &\ket{\psi_6^{(3)}}=(\ket{1}-\ket{2})(\ket{0}+\ket{2}),&\\
	&\ket{\psi_6^{(4)}}=(\ket{1}-\ket{2})(\ket{0}-\ket{2}).&
	\end{align*}
 Let
\begin{equation*}
\ket{S}=(\ket{0}+\ket{1}+\ket{2}+\ket{3})(\ket{0}+\ket{1}+\ket{2}+\ket{3})
\end{equation*}
be a \emph{stopper} state. We claim that the set \[\cU=\mathcal{B}\cup\{\ket{S}\}\setminus \{\ket{\psi_i^{(1)}}\}_{i=1}^6\] is a UPB in $\bbC^4\otimes \bbC^4$. First one can verify that $\cU$ is an ICOPB. Next the missing states $\{\ket{\psi_i^{(1)}}\}_{i=1}^6$ are not orthogonal to $\ket{S}$ but are orthogonal to all states in $\cU\setminus\{\ket{S}\}$. Then any state in $\cH_{\cU}^{\bot}$ is a linear combination of at least two of the missing states, and is orthogonal to $\ket{S}$. Assume $\ket{\psi}=a_1\ket{\psi_1^{(1)}}+a_2\ket{\psi_2^{(1)}}+a_3\ket{\psi_3^{(1)}}+a_4\ket{\psi_4^{(1)}}+a_5\ket{\psi_5^{(1)}}+a_6\ket{\psi_6^{(1)}}\in\cH_{\cU}^{\bot}$ is a product state, where at least two  coefficients are nonzero. By the correspondence between pure states and matrices,  $\ket{S}$ corresponds to the all one matrix $	J=\begin{pmatrix}
&1  &1   &1   &1            \\
&1  &1   &1   &1   \\
&1  &1   &1   &1   \\
&1  &1   &1   &1       \\
\end{pmatrix}.$
Suppose that $\ket{\psi}$ corresponds to a matrix $M=\begin{pmatrix}
&a_1  &a_1     &a_2   &a_3            \\
&a_6 &a_4    &a_6       &a_3   \\
&a_6 &a_4   &a_6     &a_3   \\
&a_5    &a_4     &a_2   &a_5       \\
\end{pmatrix}.
$
Since $\rank(M)=1$, we have $a_1=a_2=a_3=a_4=a_5=a_6\neq 0$. However $\ket{\psi}$ (resp. $M$) cannot be orthogonal to $\ket{S}$ (resp. $J$), and we have a contradiction. We have proved  that $\cU$ is a UPB of size $11$ in $\bbC^4\otimes \bbC^4$.
\end{example}

Recently it has been shown that there is no $4$-qubit UPB of size $11$ \cite{chen2018no}. In contrast, we have constructed a UPB of size $11$ in $\bbC^4\otimes\bbC^4$. This shows the  difference between the $4$-qubit system and the bipartite system  $\bbC^4\otimes\bbC^4$, because the two-qubit entangling states are allowed in the latter system.

From Example~\ref{ex:4x4upb=size11}, one may wonder what tile structures can give UPBs. By this motivation, we introduce the U-tile structures in Definition \ref{df:utile}. For a tile structure $\cT$ in $\bbC^m\otimes\bbC^n$, let  $T=\cup_{j=1}^{k}t_{i_j}$ ($k\geq 2$),  where $t_{i_j}$ is a tile. If $T$ is a sub-rectangle  of $\cT$, then $T$ is called a \emph{special rectangle} of  $\cT$. In Figure~\ref{Figure:non}, tiles $1, 2$ form a special rectangle, tiles $3, 5$ form a special rectangle, and tiles $3, 4, 5$ form a special rectangle and so on. The tile structure in Figure~\ref{Figure:example} has only one special rectangle, namely the tile structure itself $\cup_{i=1}^{6}t_i$.
For convenience, we denote $R_i$ and $C_i$ the sets of row indices and column indices of the tile $i$, respectively.
 For example, tile $1$ in Figure~\ref{Figure:non} has  row indices $0$ and   column indices $0,1$, that is, $R_1=\{0\}$ and $C_1=\{0,1\}$.  Now we are in a position to define the U-tile structure.

\begin{figure}[H]
	\centering
	\begin{tikzpicture}
	[
	box/.style={rectangle,draw=black,thick, minimum size=1cm},
	]
	
	\foreach \x in {0,1,...,3}{
		\foreach \y in {0,1,...,3}
		\node[box] at (\x,\y){};
	}
	\node[box]  at (0,3){1};
	\node[box]  at (1,3){1};
	\node[box]  at (2,3){2};
	\node[box]  at (3,3){2};
	\node[box]  at (0,2){3};
	\node[box]  at (3,2){3};
	\node[box]  at (1,2){4};
	\node[box]  at (2,2){4};
	\node[box]  at (1,1){4};
	\node[box]  at (2,1){4};
	\node[box]  at (0,1){5};
	\node[box]  at (3,1){5};
	\node[box]  at (0,0){6};
	\node[box]  at (1,0){6};
	\node[box]  at (2,0){6};
	\node[box]  at (3,0){6};
	
	\draw (-1,3) node {0};
	\draw (-1,2) node {1};
	\draw (-1,1) node {2};
	\draw (-1,0) node {3};
	\draw (0,4) node {0};
	\draw (1,4) node {1};
	\draw (2,4) node {2};
	\draw (3,4) node {3};
	\end{tikzpicture}
	\caption{Tile structure in $\bbC^4\otimes \bbC^4$.}\label{Figure:non}
\end{figure}

\begin{definition}
\label{df:utile} Given a  tile structure $\cT$, if  any special rectangle $T$ of  $\cT$ can not be partitioned into two smaller special rectangles or tiles of $\cT$,   then we call $\cT$  a \emph{U-tile structure}.
\end{definition}

Since the order among all tiles does not matter, we can always assume that a special rectangle $T=\cup_{i=1}^{k}t_{i}$ for some $k\geq 2$. By Definition~\ref{df:utile}, if $\cT$ is a U-tile structure, then both $\{R_i\}_{i=1}^{k}$ and $\{C_i\}_{i=1}^{k}$ can not  be partitioned into two parts, such that any member from one part is disjoint from all members from another part.
  The tile structure in Figure~\ref{Figure:non} is not a U-tile structure, since  the special rectangle $t_1\cup t_2$ can be partitioned into two tiles. It is easy to check that Figure~\ref{Figure:example} is a  U-tile structure, since it  has only one special rectangle $\cup_{i=1}^{6}t_i$, and $\{C_i\}_{i=1}^{6}=\{\{0,1\}, \{2\}, \{3\}, \{1\}, \{0,3\}, \{0,2\}\}$ can not be partitioned into two parts without intercrossing members. The same is for rows. In the next section, we will show that U-tile structures  correspond to UPBs.

\section{U-Tile structures and UPBs}
\label{sec:utile}

In this section, we investigate the relations between tile structures and UPBs. We give a necessary and sufficient condition for a tile structure that corresponds to a UPB in Theorem \ref{Thm:U-tile}. Then we construct some UPBs with large size by constructing U-tile structures in Propositions \ref{pro:generalization} and \ref{pro:upb42m-1}.

\begin{theorem}\label{Thm:U-tile}
	A tile structure with $s$-tiles corresponds to a UPB of size $(mn-s+1)$ in $\bbC^m\otimes \bbC^n$ if and only if this tile structure is a U-tile structure.
\end{theorem}
\begin{proof}First, we prove the sufficiency. Assume the U-tile structure  $\cT=\cup_{i=1}^{s}t_i$ with row indices $0,1,\ldots, m-1$ and column indices $0,1,\ldots, n-1$. For each tile $t_i$  with rows in $R_i=\{r_0,r_1,\dots,r_{p-1}\}$ and columns in $C_i=\{c_0,c_1,\dots,c_{q-1}\}$, we construct a set  of $pq$ orthogonal product states as follows.
For each $0\leq k\leq p-1$ and $0\leq l\leq q-1$, let
\[\ket{\phi_i^{(k,l)}}=\left(\sum_{e=0}^{p-1}w_p^{ke}\ket{r_{e}}\right)\left(\sum_{e=0}^{q-1}w_q^{le}\ket{c_{e}}\right).\]
Denote $\cB_i$ the collection of these $pq$ states given by tile $i$. Let $\ket{S}=(\sum_{e=0}^{m-1}\ket{e})(\sum_{j=0}^{n-1}\ket{j})$ be the stopper state. We claim that \[\cU=\left\{\cB_i\setminus\ket{\phi_i^{(0,0)}}\right\}_{i=1}^s\cup \left\{\ket{S}\right\}\]is a UPB of size $(mn-s+1)$ in $\bbC^m\otimes\bbC^n$.

	
The missing states are $\{\ket{\phi_i^{(0,0)}}\}_{i=1}^s$, which  are not orthogonal to $\ket{S}$ but are orthogonal to all states in $\cU\setminus\{\ket{S}\}$. Then any  state in $\cH_{\cU}^{\bot}$ must be a linear combination of the missing states (with at least two nonzero coefficients) and is orthogonal to $\ket{S}$. Assume $\ket{\psi}=\sum_{i=1}^{s}a_i\ket{\phi_i^{(0,0)}}\in\cH_{\cU}^{\bot}$ is a product state. Let $M_i$ be the corresponding $0$-$1$ matrix associated with $\ket{\phi_i^{(0,0)}}$, whose nonzero entries form  tile $i$. Then $\ket{\psi}$  corresponds to a matrix  $M=\sum_{i=1}^{s}a_iM_i$, where entries with $a_i$ form the tile $i$. Since $\rank(M)=1$, then nonzero entries of $M$ must form a special rectangle of $\cT$. Without loss of generality, let the special rectangle $T=\cup_{i=1}^{k}t_i$ for $k\geq 2$. Since $\cT$ is a U-tile structure, then all nonzero entries of $M$, that is, those entries in $T$ are the same. This is a contradiction, since $M$ is assumed to be orthogonal to the all one matrix $J$.

Now we prove the necessity by contradiction.  If $\cT$ is not a U-tile structure, then there exists a special rectangle  $T=\cup_{i=1}^{k}t_i$ with $2\leq k\leq s$, and   $\{R_i\}_{i=1}^{k}$ or $\{C_i\}_{i=1}^{k}$ can be divided into two disjoint sets without intercrossing members.   Without loss of generality, we can assume $\{C_i\}_{i=1}^{k}=\{C_i\}_{i=1}^{k'}\cup\{C_j\}_{j=k'+1}^{k}$, where $C_i$ and $C_j$ are disjoint. So we can assume $\cup_{i=1}^{k'} C_i=\{0,1,\dots,\ell-1\}$, and $\cup_{j=k'+1}^{k}C_j=\{\ell,\ell+1,\dots,h-1\}$. Now we construct a state $\ket{\psi}=\sum_{i=1}^{k}a_i\ket{\phi_i^{(0,0)}}$, where $a_i=1$ for $1\leq i\leq k'$ and $a_j=-\frac{\ell}{h-\ell}$ for $k'+1\leq j\leq k$.  That is, $\ket{\psi}$ corresponds to a matrix $M$ with nonzero entries forming a submatrix
  	$$M'=\begin{pmatrix}
 &1      &\ldots   &1  &-\frac{\ell}{h-\ell}   &\ldots   &-\frac{\ell}{h-\ell}   \\
 &1      &\ldots   &1  &-\frac{\ell}{h-\ell}   &\ldots   &-\frac{\ell}{h-\ell}   \\
 &\vdots &         &\vdots  &-\frac{\ell}{h-\ell}   &         &-\frac{\ell}{h-\ell}   \\
 &1      &\ldots   &1  &-\frac{\ell}{h-\ell}   &\ldots   &-\frac{\ell}{h-\ell}   \\
 \end{pmatrix}.$$ Then $\rank(M)=1$ and $M$ is orthogonal to $J$. It means that we can find a product state $\ket{\psi}$ in $\cH_{\cU}^{\bot}$.
   Thus $\cU$ can be extended if $\cT$ is not a  U-tile structure.
\end{proof}

In Figure~\ref{Figure:non}, since the tile structure is not a U-tile structure, it does not correspond to a UPB by Theorem~\ref{Thm:U-tile}. In fact we can find a product state $\ket{\psi}=\ket{0}(\ket{0}+\ket{1})-\ket{0}(\ket{2}+\ket{3})\in\cH_{\cU}^{\bot}$, and $\ket{\psi}$ corresponds to the matrix
	$M=\begin{pmatrix}
&1 &1   &-1  &-1            \\
&0 &0   &0   &0   \\
&0 &0   &0   &0   \\
&0 &0   &0   &0       \\
\end{pmatrix}.$

In \cite{halder2019family}, the authors gave a construction of a U-tile structure with $(2m-1)$-tiles in $\bbC^m\otimes\bbC^m$ when $m\geq 3$ is odd. They also proposed an open problem: whether this construction can be generalized for even-dimensional quantum systems? We give an affirmative answer to this question in Proposition~\ref{pro:generalization} by   constructing U-tile structures  for  arbitrary bipartite quantum systems.

\begin{proposition}\label{pro:generalization}
	There exists a UPB of size $(mn-4\fl{m-1}{2})$ in $\bbC^m\otimes\bbC^n$ for $3\leq m\leq n$.
\end{proposition}
\begin{proof}
	When $m$ is even, we can construct a U-tile structure with $(2m-3)$-tiles in Figure \ref{fig:2m-3}. 	When $m$ is odd, we can construct a U-tile structure with $(2m-1)$-tiles in Figure \ref{fig:2m-1}.
	Thus we can construct a UPB of size $(mn-4\fl{m-1}{2})$ in $\bbC^m\otimes\bbC^n$ for $3\leq m\leq n$ by Theorem~\ref{Thm:U-tile}.
\end{proof}
	{\small

		\begin{figure}[htb]
		\centering
		\begin{tikzpicture}
		[
		box/.style={rectangle,draw=black,thick, minimum size=1cm},
		]
		
		\foreach \x in {0,1,...,6}{
			\foreach \y in {0,1,...,5}
			\node[box] at (\x,\y){};
		}
		\node[box]  at (0,5){1};
		\node[box]  at (1,5){1};
		\node[box]  at (2,5){1};
		\node[box]  at (3,5){\ldots};
		\node[box]  at (4,5){1};
		\node[box]  at (5,5){1};
		\node[box]  at (6,5){2};
		\node[box]  at (6,4){2};
		\node[box]  at (6,3){\vdots};
		\node[box]  at (6,2){2};
		\node[box]  at (6,1){2};
		\node[box]  at (6,0){3};
		\node[box]  at (5,0){3};
		\node[box]  at (4,0){3};
		\node[box]  at (3,0){\ldots};
		\node[box]  at (2,0){3};
		\node[box]  at (1,0){3};
		\node[box]  at (0,0){4};
		\node[box]  at (0,1){4};
		\node[box]  at (0,2){\vdots};
		\node[box]  at (0,3){4};
		\node[box]  at (0,4){4};
		\node[box]  at (1,4){\ldots};
		\node[box]  at (2,4){\ldots};
		\node[box]  at (3,4){\ldots};
		\node[box]  at (4,4){\ldots};
		\node[box]  at (5,4){\vdots};
		\node[box]  at (5,3){\vdots};
		\node[box]  at (5,2){\vdots};
		\node[box]  at (5,1){\ldots};
		\node[box]  at (4,1){\ldots};
		\node[box]  at (3,1){\ldots};
		\node[box]  at (2,1){\ldots};
		\node[box]  at (1,1){\vdots};
		\node[box]  at (1,2){\vdots};
		\node[box]  at (1,3){\vdots};
		\node[box]  at (2,3){2$m$-3};
		\node[box]  at (3,3){\ldots};
		\node[box]  at (4,3){2$m$-3};
		\node[box]  at (4,2){2$m$-3};
		\node[box]  at (3,2){\ldots};
		\node[box]  at (2,2){2$m$-3};
		\draw (-1,5) node {0};
		\draw (-1,4) node {1};
		\draw (-1,3) node {2};							
		\draw (-1,2) node {\vdots};
		\draw (-1,1) node {$m$-2};
		\draw (-1,0) node {$m$-1};
		\draw (0,6) node {0};
		\draw (1,6) node {1};
		\draw (2,6) node {2};
		\draw (3,6) node {\ldots};
		\draw (4,6) node {$n$-3};
		\draw (5,6) node {$n$-2};
		\draw (6,6) node {$n$-1};
		\end{tikzpicture}
		\caption{A U-tile structure with (2$m$-3)-tiles in $\bbC^m\otimes \bbC^{n}$ when $m$ is even.}
	\label{fig:2m-3}	
	\end{figure}
}
	
		\begin{widetext}

	\begin{figure}[htb]
		\centering
		\begin{tikzpicture}
		[
		box/.style={rectangle,draw=black,thick, minimum size=1cm},
		]
		
		\foreach \x in {0,1,...,8}{
			\foreach \y in {0,1,...,6}
			\node[box] at (\x,\y){};
		}
		\node[box]  at (0,6){1};
		\node[box]  at (1,6){1};
		\node[box]  at (2,6){1};
		\node[box]  at (3,6){1};
		\node[box]  at (4,6){\ldots};
		\node[box]  at (5,6){1};
		\node[box]  at (6,6){1};
		\node[box]  at (7,6){1};
		\node[box]  at (8,6){2};
		\node[box]  at (8,5){2};
		\node[box]  at (8,4){2};
		\node[box]  at (8,3){\vdots};
		\node[box]  at (8,2){2};
		\node[box]  at (8,1){2};
		\node[box]  at (8,0){3};
		\node[box]  at (7,0){3};
		\node[box]  at (6,0){3};
		\node[box]  at (5,0){3};
		\node[box]  at (4,0){\ldots};
		\node[box]  at (3,0){3};
		\node[box]  at (2,0){3};
		\node[box]  at (1,0){3};
		\node[box]  at (0,0){4};
		\node[box]  at (0,1){4};
		\node[box]  at (0,2){4};
		\node[box]  at (0,3){\vdots};		
		\node[box]  at (0,4){4};
		\node[box]  at (0,5){4};	
		\node[box]  at (1,5){\dots};
		\node[box]  at (2,5){\dots};		
		\node[box]  at (3,5){\dots};		
		\node[box]  at (4,5){\dots};		
		\node[box]  at (5,5){\dots};		
		\node[box]  at (6,5){\dots};		
		\node[box]  at (7,5){\vdots};
		\node[box]  at (7,4){\vdots};		
		\node[box]  at (7,3){\vdots};		
		\node[box]  at (7,2){\vdots};		
		\node[box]  at (7,1){\ldots};
		\node[box]  at (6,1){\ldots};		
		\node[box]  at (5,1){\ldots};		
		\node[box]  at (4,1){\ldots};		
		\node[box]  at (3,1){\ldots};		
		\node[box]  at (2,1){\ldots};		
		\node[box]  at (1,1){\vdots};
		\node[box]  at (1,2){\vdots};	
		\node[box]  at (1,3){\vdots};
		\node[box]  at (1,4){\vdots};	
		\node[box]  at (2,4){2$m$-5};
		\node[box]  at (3,4){\ldots};
		\node[box]  at (4,4){2$m$-5};	
		\node[box]  at (5,4){2$m$-5};
		\node[box]  at (6,4){2$m$-4};		
		\node[box]  at (6,3){2$m$-4};	
		\node[box]  at (6,2){2$m$-3};	
		\node[box]  at (5,2){2$m$-3};	
		\node[box]  at (4,2){\ldots};	
		\node[box]  at (3,2){2$m$-3};	
		\node[box]  at (2,2){2$m$-2};	
		\node[box]  at (2,3){2$m$-2};	
		\node[box]  at (3,3){2$m$-1};
		\node[box]  at (4,3){\ldots};
		\node[box]  at (5,3){2$m$-1};	
		\draw (-1,6) node {0};
		\draw (-1,5) node {1};
		\draw (-1,4) node {2};							
		\draw (-1,3) node {\vdots};
		\draw (-1,2) node {$m$-3};
		\draw (-1,1) node {$m$-2};
		\draw (-1,0) node {$m$-1};
		\draw (0,7) node {0};
		\draw (1,7) node {1};
		\draw (2,7) node {2};
		\draw (3,7) node {3};
		\draw (4,7) node {\ldots};
		\draw (5,7) node {$n$-4};
		\draw (6,7) node {$n$-3};
		\draw (7,7) node {$n$-2};
		\draw (8,7) node {$n$-1};
		\end{tikzpicture}
		\caption{A U-tile structure with $(2m-1)$-tiles in $\bbC^m\otimes \bbC^{n}$ when $m$ is odd.}
		\label{fig:2m-1}
	\end{figure}

\end{widetext}

\begin{proposition}\label{pro:upb42m-1}
		There is a UPB of size $(mn-k)$  in $\bbC^m\otimes\bbC^n$ for $ 4\leq k\leq 2m-1$ and  $4\leq m\leq n$. $F(m,n)=mn-4$ for $3\leq m\leq n$, where $F(m,n)$ is the maximum size of UPBs in $\bbC^m\otimes\bbC^n$.
\end{proposition}

The proof of Proposition \ref{pro:upb42m-1} is given in Appendix~\ref{app:proof=prop5}.  In Proposition~\ref{pro:upb42m-1}, we also give a construction of a U-tile structure with $(2m-1)$-tiles in $\bbC^m\otimes\bbC^m$ when $m$ is odd. But this U-tile structure has only one special rectangle, which is different from the U-tile structure in \cite{halder2019family} that has at least two special rectangles when $m \geq 5$.

UPBs can be used to construct PPT entangled states \cite{bennett1999unextendible}.  From Proposition~\ref{pro:upb42m-1}, we can construct a UPB $\{\ket{\psi_i}\}_{i=1}^{mn-k}$ for $ 4\leq k\leq 2m-1$ and $4\leq m\leq n$, then  $\rho=I-\sum_{i=1}^{mn-k}\ket{\psi_i}\bra{\psi_i}$ is a rank-$k$ PPT entangled state. The bipartite state is either separable or entangled. Determining whether a state is entangled is an NP-hard problem, namely the separability problem. It has been shown that on the bipartite Hilbert space $\bbC^2\otimes\bbC^2$ and $\bbC^2\otimes\bbC^3$, the state $\rho$ is separable if and only if it is a positive-partial-transpose (PPT) state \cite{peres1996,horodecki2001separability}.  For the systems of high dimensions, there exist PPT entangled states \cite{horodecki1997}. PPT entangled states represent the so-called bound entangled states from which no pure entanglement can be distilled under LOCC  \cite{horodecki1997,horodecki1999bound}. It is also related to the long-standing conjecture wondering whether there exists a negative-partial-transpose bound entangled state \cite{horodecki2020five}. Therefore, our construction of PPT entangled states shows novel understanding of these problems.


\section{Application: Local distinguishability of  UPBs by entanglement resource}
\label{sec:app}

In this section, we provide a method of locally distinguishing UPBs constructed in Proposition~\ref{pro:generalization}  assisted by entanglement, because UPBs cannot be distinguished perfectly by LOCC alone \cite{cohen2008understanding}.  When $m=n\geq 3$ are odd, the UPB constructed from Proposition~\ref{pro:generalization} can be perfectly distinguished by LOCC with a $\lceil\frac{m}{2}\rceil\otimes\lceil\frac{m}{2}\rceil$ maximally entangled state \cite{zhang2020locally}. We will prove that any UPB from Proposition~\ref{pro:generalization} can be perfectly distinguished by LOCC with a $\lceil\frac{m}{2}\rceil\otimes \lceil\frac{m}{2}\rceil$  maximally entangled state in Theorem~\ref{Thm:mevendis}.

We begin by showing the special case $4=m\leq n$ in Lemma~\ref{lem:disti44}. For this purpose, we demonstrate the  UPB of size $4n-4$ in $\bbC^4\otimes\bbC^n$ constructed by the U-tile structure in Proposition~\ref{pro:generalization} as follows.
\begin{align*}
\notag
\ket{\psi_i}=&\ket{0}\left(\sum_{j=0}^{n-2}w_{n-1}^{ij}\ket{j}\right), 1\leq i\leq n-2,\\
\notag
\ket{\psi_{i+n-2}}=&\left(\sum_{j=0}^{2}w_{3}^{ij}\ket{j}\right)\ket{n-1}, 1\leq i\leq 2,\\
\notag
\ket{\psi_{i+n}}=&\ket{3}\left(\sum_{j=1}^{n-1}w_{n-1}^{ij}\ket{j}\right), 1\leq i\leq n-2,\\
\notag
\ket{\psi_{i+2n-2}}=&\left(\sum_{j=1}^{3}w_{3}^{ij}\ket{j}\right)\ket{0}, 1\leq i\leq 2,
\end{align*}
\begin{align}\label{Eq:12sizeUPB44}
\notag
\ket{\psi_{i+2n}}=&\left(\ket{1}+\ket{2}\right)\left(\sum_{j=1}^{n-2}w_{n-2}^{ij}\ket{j}\right), 1\leq i\leq n-3,\\
\notag
\ket{\psi_{i+3n-2}}=&\left(\ket{1}-\ket{2}\right)\left(\sum_{j=1}^{n-2}w_{n-2}^{ij}\ket{j}\right), 0\leq i\leq n-3,\\
\notag
\ket{S}=&\left(\ket{0}+\ket{1}+\ket{2}+\ket{3}\right)\\
&\left(\ket{0}+\ket{1}+\ldots+\ket{n-1}\right).
\end{align}
We show that the above states can be perfectly
distinguished by LOCC  assisted with entanglement.

\begin{lemma}\label{lem:disti44}
	The  UPB of Eqs. (\ref{Eq:12sizeUPB44}) can be perfectly
	distinguished by LOCC with a $2\otimes2$ maximally entangled
	state.
\end{lemma}
\begin{proof}
Let Alice and Bob share a $2\otimes2$ maximally entangled
state $\ket{\psi}_{ab}=\ket{00}+\ket{11}$. Let $\ket{\psi_i'}=\ket{\psi_i}_{AB}\ket{\psi}_{ab}$ for $1\leq i\leq 4n-5$ and $\ket{S'}=\ket{S}_{AB}\ket{\psi}_{ab}$. Then  Alice performs a two-outcome measurement on each  of the $4n-4$ states $\ket{\psi_i'}$ and $\ket{S'}$, each outcome corresponding to a rank-4 projector:
\begin{align*}
&A_1=\ket{00}_{Aa}\bra{00}+\ket{10}_{Aa}\bra{10}+\ket{20}_{Aa}\bra{20}+\ket{31}_{Aa}\bra{31};\\
&A_2=\ket{01}_{Aa}\bra{01}+\ket{11}_{Aa}\bra{11}+\ket{21}_{Aa}\bra{21}+\ket{30}_{Aa}\bra{30}.
\end{align*}
For operating with  $A_1$  on systems $Aa$, each of the initial states is transformed into:
\begin{align}\label{Eq:A1acts}
\notag
\ket{\phi_i}=&\ket{\psi_i}\ket{00}, 1\leq i\leq n \ \text{and} \ 2n+1\leq i\leq 4n-5,\\
\notag
\ket{\phi_i}=&\ket{\psi_i}\ket{11}, n+1\leq i \leq 2n-2,\\
\notag
\ket{\phi_i}=&\left(\sum_{j=1}^{2}w_3^{ij}\ket{j}\right)\ket{0}\ket{00}+w_3^{3i}\ket{3}\ket{0}\ket{11}, i=2n-1,2n,\\
\notag
\ket{S}\rightarrow&\left(\ket{0}+\ket{1}+\ket{2}\right)\left(\ket{0}+\ket{1}+\ldots+\ket{n-1}\right)\ket{00}\\
&+\ket{3}\left(\ket{0}+\ket{1}+\ldots+\ket{n-1}\right)\ket{11}.
\end{align}  	
We only need to consider the operator $A_1$, since operating $A_2$  on systems $Aa$  generates new states which differ from the states in Eqs. (\ref{Eq:A1acts}) only by ancillary systems $\ket{00}_{ab}\rightarrow\ket{11}_{ab}$ and $ \ket{11}_{ab}\rightarrow\ket{00}_{ab}$.

Now, we show the local distinguishability of the states in Eqs. (\ref{Eq:A1acts}).  Bob makes an $(n+1)$-outcome projective measurement, where the first $n-1$ projectors are $B_i=\left(\sum_{j=1}^{n-1}w_{n-1}^{ij}\ket{j}\right)_B\left(\sum_{j=1}^{n-1}w_{n-1}^{ij}\bra{j}\right)\otimes\ket{1}_b\bra{1}$, $1\leq i\leq n-1$. For each $B_i$, $1\leq i\leq n-2$,  the only remaining possibility
is $\ket{\phi_{i+n}}$,  which has thus been successfully identified. In the same way,  Bob can identify  $\ket{S}$ by $B_{n-1}$.

Then Bob uses the $n$th projector $B_n=\ket{n-1}_B\bra{n-1}\otimes\ket{0}_b\bra{0}$. It leaves $\ket{\phi_{i}}$, $i=n-1,n$, and $\ket{S}\rightarrow (\ket{0}+\ket{1}+\ket{2})\ket{n-1}\ket{00}$.  Now Bob has the
same state in his own party and Alice has orthogonal states. Thus, Alice can
distinguish these states.

Bob's last outcome is a projector  $B_{n+1}=I-B_1-B_2-\ldots-B_n$. It leaves $\ket{\phi_{i}}$, $1\leq i\leq n-2$ and $2n-1\leq i\leq 4n-5$, and $\ket{S}\rightarrow(\ket{0}+\ket{1}+\ket{2})\left(\sum_{j=0}^{n-2}\ket{j}\right)\ket{00}+\ket{3}\ket{0}\ket{11}$.
Then, Alice uses the projector $A_{n+1,1}=\ket{0}_A\bra{0}\otimes\ket{0}_a\bra{0}$,  leaving $\ket{\phi_{i}}$, $1\leq i\leq n-2$ and $\ket{S}\rightarrow\ket{0}\left(\sum_{j=0}^{n-2}\ket{j}\right)\ket{00}$, which can be easily
distinguished by Bob. When Alice uses the projector $A_{n+1,2}=I-A_{n+1,1}$, it leaves  $\ket{\phi_{i}}$, $2n-1\leq i\leq 4n-5$, and $\ket{S}\rightarrow(\ket{1}+\ket{2})\left(\sum_{j=0}^{n-2}\ket{j}\right)\ket{00}+\ket{3}\ket{0}\ket{11}$.
 Then, Bob uses $B_{n+1,2,1}=\ket{0}_B\bra{0}\otimes(\ket{0}_b\bra{0}+\ket{1}_b\bra{1})$, leaving $\ket{\phi_{i}}$, $i=2n-1,2n$, and $\ket{S}\rightarrow (\ket{1}+\ket{2})\ket{0}\ket{00}+\ket{3}\ket{0}\ket{11}$.  Bob makes a projective measurement on system $b$ by projecting $\ket{0}_b+\ket{1}_b$ and  $\ket{0}_b-\ket{1}_b$. Then every projector can get the same state in Bob's party and Alice has orthogonal states. Thus, Alice can distinguish these states. When Bob uses projector $B_{n+1,2,2}=I-B_{n+1,2,1}$, it leaves $\ket{\phi_{i}}$, $2n+1\leq i\leq 4n-5$ and $\ket{S''}= (\ket{1}+\ket{2})(\ket{1}+\ket{2}+\ldots+\ket{n-2})\ket{00}$. Alice uses projector $A_{n+1,2,2,1}=(\ket{1}+\ket{2})_A(\bra{1}+\bra{2})$,  leaving $\ket{\phi_{i}}$, $2n+1\leq i\leq 3n-3$, and $\ket{S''}$, which can be easily distinguished by Bob. Then Alice uses projector $A_{n+1,2,2,2}=I-A_{n+1,2,2,1}$, which leaves $\ket{\phi_{i}}$, $3n-2\leq i\leq 4n-5$. But Bob can easily distinguish these states.

Thus,  the states  in Eqs. (\ref{Eq:12sizeUPB44}) can be perfectly distinguished by  LOCC with a $2\otimes2$ maximally entangled
state through our protocol.
\end{proof}

Next, we consider the general UPBs in Proposition~\ref{pro:generalization}. When $4\leq m\leq n$ and $m$ is even, we can construct an UPB of size $mn-2m+4$ in $\bbC^m\otimes\bbC^n$ using the U-tile structure in Proposition~\ref{pro:generalization} as follows. For convenience, denote   $\iota\triangleq \frac{m}{2}$.
\begin{align*}
\notag
\ket{\psi_i}=&\ket{0}\left(\sum_{j=0}^{n-2}w_{n-1}^{ij}\ket{j}\right), 1\leq i\leq n-2,\\
\notag
\ket{\psi_{i+n-2}}=&\left(\sum_{j=0}^{m-2}w_{m-1}^{ij}\ket{j}\right)\ket{n-1}, 1\leq i\leq m-2,\\
\notag
\ket{\psi_{i+m+n-4}}=&\ket{m-1}\left(\sum_{j=1}^{n-1}w_{n-1}^{ij}\ket{j}\right), 1\leq i\leq n-2,\\
\notag
\ket{\psi_{i+m+2n-6}}=&\left(\sum_{j=1}^{m-1}w_{m-1}^{ij}\ket{j}\right)\ket{0}, 1\leq i\leq m-2,\\
\notag
\ket{\psi_{i+2m+2n-8}}=&\ket{1}\left(\sum_{j=1}^{n-3}w_{n-3}^{ij}\ket{j}\right), 1\leq i\leq n-4,\\
\notag
\ket{\psi_{i+2m+3n-12}}=&\left(\sum_{j=1}^{m-3}w_{m-3}^{ij}\ket{j}\right)\ket{n-2}, 1\leq i\leq m-4,\\
\notag
\ket{\psi_{i+3m+3n-16}}=&\ket{m-2}\left(\sum_{j=2}^{n-2}w_{n-3}^{ij}\ket{j}\right), 1\leq i\leq n-4,
\end{align*}
\begin{align}\label{Eq:2m3sizeUPBmm}
\notag
\ket{\psi_{i+3m+4n-20}}=&\left(\sum_{j=2}^{m-2}w_{m-3}^{ij}\ket{j}\right)\ket{1}, 1\leq i\leq m-4,\\
\notag
&\ldots\\
\notag
\ket{\psi_{mn-2n+i}}=&\left(\ket{\iota-1}+\ket{\iota}\right)\left(\sum_{j=\iota-1}^{n-\iota}w_{n-m+2}^{ij}\ket{j}\right),\\
\notag
 &1\leq i\leq n-m+1, \\
\notag
\ket{\psi_{mn-n-m+2+i}}=&\left(\ket{\iota-1}-\ket{\iota}\right)\left(\sum_{j=\iota-1}^{n-\iota}w_{n-m+2}^{ij}\ket{j}\right),\\
\notag
& 0\leq i\leq n-m+1, \\
\notag
\ket{S}=&\left(\ket{0}+\ket{1}+\ldots+\ket{m-1}\right)\\
&\left(\ket{0}+\ket{1}+\ldots+\ket{n-1}\right).
\end{align}

We first show that the above states can be perfectly distinguished by LOCC with an $\iota\otimes \iota$ maximally entangled state in Theorem \ref{Thm:mevendis}. Then we consider  the  UPBs in Proposition~\ref{pro:generalization} for $m$ is odd in Theorem \ref{Thm:mevendis}.
%

\begin{theorem}\label{Thm:mevendis}
  The UPB constructed in Proposition~\ref{pro:generalization}  can be perfectly distinguished by LOCC with a $\lceil\frac{m}{2}\rceil\otimes \lceil\frac{m}{2}\rceil$ maximally entangled state.
\end{theorem}
\begin{proof} Let $m\geq 4$ be even.
	We prove it by induction on $m$.
	When $m=4$, we have proved the statement in Lemma~\ref{lem:disti44}.  When $k=m-2$, assume the states in Eqs. (\ref{Eq:2m3sizeUPBmm}) can be locally distinguished with an  $(\iota-1)\otimes(\iota-1)$ maximally entangled state for any $n\geq m-2$. We only need to show when $k=m$, Eqs. (\ref{Eq:2m3sizeUPBmm}) can be locally distinguished with an  $\iota\otimes \iota$ maximally entangled state for any $n\geq m$.
	let Alice and Bob share an $\iota\otimes\iota$ maximally entangled
	state $\ket{\psi}_{ab}=\sum_{j=0}^{\iota-1}\ket{jj}$. Let $\ket{\psi_i'}=\ket{\psi_i}_{AB}\ket{\psi}_{ab}$ for $1\leq i\leq mn-2m+3$ and $\ket{S'}=\ket{S}_{AB}\ket{\psi}_{ab}$. Then  Alice performs an $\iota$-outcome measurement on each of these $(mn-2m+4)$ states, each outcome corresponding to a rank-$m$ projector:
	\begin{align*}
	A_1=&\ket{00}_{Aa}\bra{00}+\ket{10}_{Aa}\bra{10}+\ldots+\ket{\iota0}_{Aa}\bra{\iota0}\\
	&+\ket{(\iota+1)1}_{Aa}\bra{(\iota+1)1}+\ldots\\
	&+\ket{(m-1)(\iota-1)}_{Aa}\bra{(m-1)(\iota-1)};\\
    A_{i}=&\sum_{j=0}^{\iota-1}\ket{j(i-1)}_{Aa}\bra{j(i-1)}+\\
   &\sum_{j=0}^{\iota-1}\ket{(\iota+j)(j+i-1)}_{Aa}\bra{(\iota+j)(j+i-1)},
	\end{align*}
	for $2\leq i\leq \iota$. Here the additions in system $a$ are modulo $\iota$. Operating  $A_1$  on systems $Aa$, each of the initial states is transformed into:
	\begin{align}\label{Eq:mA1acts}
	\notag
	\ket{\phi_i}=&\ket{\psi_i}\ket{00}, 1\leq i\leq n-2,\\
	\notag
	\ket{\phi_{i+n-2}}=&\left(\sum_{j=0}^{\iota-1}w_{m-1}^{ij}\ket{j}\right)\ket{n-1}\ket{00}+\\
	\notag
	&\sum_{j=0}^{\iota-2}w_{m-1}^{i(\iota+j)}\ket{\iota+j}\ket{n-1}\ket{jj},\\
	\notag
	&1\leq i\leq m-2,\\
	\notag
	\ket{\phi_{i+m+n-4}}=&\ket{\psi_{i+m+n-4}}\ket{(\iota-1)(\iota-1)},\\
	\notag
	 &1\leq i\leq n-2,\\
%
	\notag \ket{\phi_{i+m+2n-6}}=&\left(\sum_{j=1}^{\iota-1}w_{m-1}^{ij}\ket{j}\right)\ket{0}\ket{00}+\\
	&\sum_{j=0}^{\iota-1}w_{m-1}^{i(\iota+j)}\ket{\iota+j}\ket{0}\ket{jj},\\
	\notag	
	&1\leq i\leq m-2,\\
	\notag
	&\ldots\\
	\notag
	\ket{S}\rightarrow&\left(\sum_{j=0}^{\iota-1}\ket{j}\right)\left(\sum_{e=0}^{n-1}\ket{e}\right)\ket{00}+\\
&\sum_{j=0}^{\iota-1}\ket{\iota+j}\left(\sum_{e=0}^{n-1}\ket{e}\right)\ket{jj}.
	\end{align}  	
	Similarly, we only need to consider the operator $A_1$.
	
	Now, we show the local distinguishability of the states in Eqs. (\ref{Eq:mA1acts}). Similar to Lemma~\ref{lem:disti44}, Bob can identify   $\ket{\phi_{i+m+n-4}}$, $1\leq i\leq n-2$, and $\ket{S}$,  by projectors $B_i=\left(\sum_{j=1}^{n-1}w_{n-1}^{ij}\ket{j}\right)_B\left(\sum_{j=1}^{n-1}w_{n-1}^{ij}\bra{j}\right)\otimes\ket{\iota-1}_{b}\bra{\iota-1}$, $1\leq i\leq n-1$.
	
	Then Bob uses  the $n$th projector $B_{n}=\ket{n-1}_B\bra{n-1}\otimes\left(\sum_{j=0}^{\iota-2}\ket{j}_b\bra{j}\right)$. It leaves $\ket{\phi_{i+n-2}}$, $1\leq i\leq m-2$, $\ket{S}\rightarrow  \left(\sum_{j=0}^{\iota-1}\ket{j}\right)\ket{n-1}\ket{00}+\sum_{j=0}^{\iota-2}\ket{\iota+j}\ket{n-1}\ket{jj}$. Then Bob makes a projective measurement on system $b$ by projecting $\sum_{i=0}^{\iota-2}w_{\iota-1}^{ij}\ket{i}_b$, $0\leq j\leq \iota-2$, and gets the same states in Bob's party. Thus, Alice can
	distinguish these states.
	
	Bob's last out come is a projector  $B_{n+1}=I-B_1-B_2-\ldots-B_n$. It leaves $\ket{\phi_{i}}$, $1\leq i\leq n-2$ and $m+2n-5\leq i\leq mn-2m+3$ and 	$\ket{S}\rightarrow \ket{S'}=\left(\sum_{j=0}^{\iota-1}\ket{j}\right)\left(\sum_{e=0}^{n-2}\ket{e}\right)\ket{00}+
\sum_{j=0}^{\iota-2}\ket{\iota+j}\left(\sum_{e=0}^{n-2}\ket{e}\right)\ket{jj}+\ket{m-1}\ket{0}\ket{(\iota-1)(\iota-1)}$.
 Then Alice uses a projector $A_{n+1,1}=\ket{0}_A\bra{0}\otimes\ket{0}_a\bra{0}$, and leaves $\ket{\phi_{i}}$, $1\leq i\leq n-2$ and $\ket{S}\rightarrow\ket{0}\left(\sum_{e=0}^{n-2}\ket{e}\right)\ket{00}$, which can be easily
	distinguished by Bob. When Alice uses the projector $A_{n+1,2}=I-A_{n+1,1}$, it leaves  $\ket{\phi_{i}}$, $m+2n-5\leq i\leq mn-2m+3$, and $\ket{S}\rightarrow  \ket{S'}- \ket{0}\left(\sum_{e=0}^{n-2}\ket{e}\right)\ket{00}$. Then, Bob uses a projector $B_{n+1,2,1}=\ket{0}_B\bra{0}\otimes\left(\sum_{j=0}^{\iota-1}\ket{j}_b\bra{j}\right)$, leaves $\ket{\phi_{i}}$, $m+2n-5\leq i\leq 2m+2n-8$, and $\ket{S}\rightarrow \left(\sum_{j=1}^{\iota-1}\ket{j}\right)\ket{0}\ket{00}+\sum_{j=0}^{\iota-1}\ket{\iota+j}\ket{0}\ket{jj}$.
Then Bob makes a projective measurement on system $b$ by projecting $\sum_{i=0}^{\iota-1}w_{\iota}^{ij}\ket{i}_b$, $0\leq j\leq \iota-1$. Then every projector can get the same state in Bob's party. Thus, Alice can distinguish these states. When Bob uses the projector $B_{n+1,2,2}=I-B_{n+1,2,1}$, it leaves $\ket{\phi_{i}}$, $2m+2n-7\leq i\leq mn-2m+3$ and $\ket{S}\rightarrow
\left(\sum_{j=1}^{\iota-1}\ket{j}\right)\left(\sum_{e=1}^{n-2}\ket{e}\right)\ket{00}+\sum_{j=0}^{\iota-2}\ket{\iota+j}\left(\sum_{e=1}^{n-2}\ket{e}\right)\ket{jj}$.
By induction hypothesis, these states in $\bbC^{m-2}\otimes \bbC^{n-2}$ are locally
	distinguishable.
	
	Thus, the states in Eqs. (\ref{Eq:2m3sizeUPBmm}) can be perfectly distinguished by  LOCC with an $\iota\otimes \iota$ maximally entangled
	state through our protocol.
	
	In Ref. \cite{zhang2020locally}, the authors showed that  when $m=n\geq 3$ are odd, the UPB constructed from Proposition~\ref{pro:generalization} can be perfectly distinguished by LOCC with a $\lceil\frac{m}{2}\rceil\otimes\lceil\frac{m}{2}\rceil$ maximally entangled states. Applying the similar argument as above to odd $m$, we can show that the result is true for any UPB in Proposition~\ref{pro:generalization} when $m\leq n$ .
%
\end{proof}

Ref. \cite{cohen2008understanding} has shown that the Gentiles2 UPB can be perfectly distinguished by LOCC with a $\lceil\frac{m}{2}\rceil\otimes \lceil\frac{m}{2}\rceil$ maximally entangled state. They also wonder whether other types of UPBs can be locally distinguished
by efficiently using entanglement resource. In Proposition~\ref{pro:generalization} we have constructed a novel type of UPBs that can be perfectly distinguished by LOCC with a $\lceil\frac{m}{2}\rceil\otimes \lceil\frac{m}{2}\rceil$ maximally entangled state. Our UPB has different size from that of Gentiles2 UPB. We  conjecture that every UPB may be distinguished in this way. Our results also show how to use entanglement efficiently.

\section{Conclusion}
\label{sec:con}

We showed that a tile structure gives a UPB if and only if it is a U-tile structure, and constructed UPBs of large size by constructing U-tile structures. We also proved that some UPBs in $\bbC^m\otimes \bbC^n$  can be perfectly distinguished by LOCC assisted with a $\lceil\frac{m}{2}\rceil\otimes\lceil\frac{m}{2}\rceil$ maximally entangled state. Our future work is to give more constructions of U-tile structures, and find  the  maximum number of tiles in a U-tile structure of size $m\times n$. It is also meaningful to extend the U-tile property to multipartite systems.

\section*{Acknowledgements}

FS and XZ were supported by NSFC under Grant No. 11771419,  the Fundamental Research Funds for the Central Universities,	and Anhui Initiative in Quantum Information Technologies under Grant No. AHY150200. LC was supported by the  NNSF of China (Grant No. 11871089), and the Fundamental Research Funds for the Central Universities (Grant Nos. KG12080401 and ZG216S1902).

\appendix
\begin{widetext}
\section{Proof of Proposition~\ref{pro:upb42m-1}}

\label{app:proof=prop5}

First, we construct a UPB of size $(m^2-k)$ in $\bbC^m\otimes\bbC^m$ for $4\leq k\leq 2m-1$ and $m\geq 4$.
	By Theorem~\ref{Thm:U-tile}, we only need to construct a U-tile structure with $t$-tiles in $\bbC^m\otimes\bbC^m$ for $m\geq 4$ and $ 5\leq t\leq 2m$.
	When $m=4$, we can construct U-tile structures with $5,6,7,8$-tiles in $\bbC^4\otimes\bbC^4$ in Figure \ref{fig:5678=4x4}.
	\begin{figure}[H]
		\centering
		\begin{tikzpicture}
		[
		box/.style={rectangle,draw=black,thick, minimum size=1cm},
		]
		
		\foreach \x in {0,1,...,3}{
			\foreach \y in {0,1,...,3}
			\node[box] at (\x,\y){};
		}
		\node[box]  at (0,3){1};
		\node[box]  at (1,3){1};
		\node[box]  at (2,3){1};
		\node[box]  at (3,3){2};
		\node[box]  at (3,2){2};
		\node[box]  at (3,1){2};
		\node[box]  at (3,0){3};
		\node[box]  at (2,0){3};
		\node[box]  at (1,0){3};
		\node[box]  at (0,0){4};
		\node[box]  at (0,1){4};
		\node[box]  at (0,2){4};
		\node[box]  at (1,1){5};
		\node[box]  at (1,2){5};
		\node[box]  at (2,1){5};
		\node[box]  at (2,2){5};	
		\end{tikzpicture}
		\quad
		\begin{tikzpicture}
		[
		box/.style={rectangle,draw=black,thick, minimum size=1cm},
		]
		
		\foreach \x in {0,1,...,3}{
			\foreach \y in {0,1,...,3}
			\node[box] at (\x,\y){};
		}
		\node[box]  at (0,3){1};
		\node[box]  at (1,3){1};
		\node[box]  at (2,3){6};
		\node[box]  at (3,3){2};
		\node[box]  at (3,2){2};
		\node[box]  at (3,1){2};
		\node[box]  at (3,0){3};
		\node[box]  at (2,0){6};
		\node[box]  at (1,0){3};
		\node[box]  at (0,0){4};
		\node[box]  at (0,1){4};
		\node[box]  at (0,2){4};
		\node[box]  at (1,1){5};
		\node[box]  at (1,2){5};
		\node[box]  at (2,1){5};
		\node[box]  at (2,2){5};	
		\end{tikzpicture}
		\quad
		\begin{tikzpicture}
		[
		box/.style={rectangle,draw=black,thick, minimum size=1cm},
		]
		
		\foreach \x in {0,1,...,3}{
			\foreach \y in {0,1,...,3}
			\node[box] at (\x,\y){};
		}
		\node[box]  at (0,3){1};
		\node[box]  at (1,3){1};
		\node[box]  at (2,3){6};
		\node[box]  at (3,3){2};
		\node[box]  at (3,2){7};
		\node[box]  at (3,1){2};
		\node[box]  at (3,0){3};
		\node[box]  at (2,0){6};
		\node[box]  at (1,0){3};
		\node[box]  at (0,0){4};
		\node[box]  at (0,1){4};
		\node[box]  at (0,2){7};
		\node[box]  at (1,1){5};
		\node[box]  at (1,2){5};
		\node[box]  at (2,1){5};
		\node[box]  at (2,2){5};	
		\end{tikzpicture}
		\quad
		\begin{tikzpicture}
		[
		box/.style={rectangle,draw=black,thick, minimum size=1cm},
		]
		
		\foreach \x in {0,1,...,3}{
			\foreach \y in {0,1,...,3}
			\node[box] at (\x,\y){};
		}
		\node[box]  at (0,3){1};
		\node[box]  at (1,3){1};
		\node[box]  at (1,2){2};
		\node[box]  at (2,2){2};
		\node[box]  at (2,1){3};
		\node[box]  at (3,1){3};
		\node[box]  at (3,0){4};
		\node[box]  at (0,0){4};
		\node[box]  at (0,2){5};
		\node[box]  at (0,1){5};
		\node[box]  at (1,1){6};
		\node[box]  at (1,0){6};
		\node[box]  at (2,3){7};
		\node[box]  at (2,0){7};
		\node[box]  at (3,3){8};
		\node[box]  at (3,2){8};	
		\end{tikzpicture}
		\caption{U-tile structures with $5,6,7,8$-tiles in $\bbC^4\otimes \bbC^{4}$.}
		\label{fig:5678=4x4}
	\end{figure}
	When $m=5$, we illustrate our construction in Figures \ref{fig:567=5x5-1} and \ref{fig:8910=5x5}. In particular we can construct U-tile structures with $5,6,7$-tiles in $\bbC^5\otimes \bbC^{5}$ based on  U-tile structures with $5,6,7$-tiles in $\bbC^4\otimes\bbC^4$; and we can construct U-tile structures with $8,9,10$-tiles in $\bbC^5\otimes \bbC^{5}$ based on the U-tile structure with $8$-tiles in $\bbC^4\otimes\bbC^4$. We append a row and a column on the top and right of the  U-tile structures in $\bbC^4\otimes\bbC^4$.
	\begin{figure}[H]
		\centering
		\begin{tikzpicture}
		[
		box/.style={rectangle,draw=black,thick, minimum size=1cm},
		]
		
		\foreach \x in {0,1,...,4}{
			\foreach \y in {0,1,...,4}
			\node[box] at (\x,\y){};
		}
		\node[box]  at (0,3){1};
		\node[box]  at (1,3){1};
		\node[box]  at (2,3){1};
		\node[box]  at (3,3){2};
		\node[box]  at (3,2){2};
		\node[box]  at (3,1){2};
		\node[box]  at (3,0){3};
		\node[box]  at (2,0){3};
		\node[box]  at (1,0){3};
		\node[box]  at (0,0){4};
		\node[box]  at (0,1){4};
		\node[box]  at (0,2){4};
		\node[box]  at (1,1){5};
		\node[box]  at (1,2){5};
		\node[box]  at (2,1){5};
		\node[box]  at (2,2){5};
		\node[box] at (0,4){1};
		\node[box] at (1,4){1};
		\node[box] at (2,4){1};
		\node[box] at (3,4){2};
		\node[box] at (4,4){2};
		\node[box] at (4,3){2};
		\node[box] at (4,2){2};
		\node[box] at (4,1){2};
		\node[box] at (4,0){3};
		\end{tikzpicture}
		\quad
		\begin{tikzpicture}
		[
		box/.style={rectangle,draw=black,thick, minimum size=1cm},
		]
		
		\foreach \x in {0,1,...,4}{
			\foreach \y in {0,1,...,4}
			\node[box] at (\x,\y){};
		}
		\node[box]  at (0,3){1};
		\node[box]  at (1,3){1};
		\node[box]  at (2,3){6};
		\node[box]  at (3,3){2};
		\node[box]  at (3,2){2};
		\node[box]  at (3,1){2};
		\node[box]  at (3,0){3};
		\node[box]  at (2,0){6};
		\node[box]  at (1,0){3};
		\node[box]  at (0,0){4};
		\node[box]  at (0,1){4};
		\node[box]  at (0,2){4};
		\node[box]  at (1,1){5};
		\node[box]  at (1,2){5};
		\node[box]  at (2,1){5};
		\node[box]  at (2,2){5};
		\node[box] at (0,4){1};
		\node[box] at (1,4){1};
		\node[box] at (2,4){6};
		\node[box] at (3,4){2};
		\node[box,] at (4,4){2};
		\node[box] at (4,3){2};
		\node[box] at (4,2){2};
		\node[box,] at (4,1){2};
		\node[box] at (4,0){3};	
		\end{tikzpicture}
		\quad
		\begin{tikzpicture}
		[
		box/.style={rectangle,draw=black,thick, minimum size=1cm},
		]
		
		\foreach \x in {0,1,...,4}{
			\foreach \y in {0,1,...,4}
			\node[box] at (\x,\y){};
		}
		\node[box]  at (0,3){1};
		\node[box]  at (1,3){1};
		\node[box]  at (2,3){6};
		\node[box]  at (3,3){2};
		\node[box]  at (3,2){7};
		\node[box]  at (3,1){2};
		\node[box]  at (3,0){3};
		\node[box]  at (2,0){6};
		\node[box]  at (1,0){3};
		\node[box]  at (0,0){4};
		\node[box]  at (0,1){4};
		\node[box]  at (0,2){7};
		\node[box]  at (1,1){5};
		\node[box]  at (1,2){5};
		\node[box]  at (2,1){5};
		\node[box]  at (2,2){5};
		\node[box] at (0,4){1};
		\node[box] at (1,4){1};
		\node[box] at (2,4){6};
		\node[box] at (3,4){2};
		\node[box] at (4,4){2};
		\node[box] at (4,3){2};
		\node[box] at (4,2){7};
		\node[box] at (4,1){2};
		\node[box] at (4,0){3};	
		\end{tikzpicture}
		\caption{U-tile structures with $5,6,7$-tiles in $\bbC^5\otimes \bbC^{5}$ based on U-tile structures with $5,6,7$-tiles in $\bbC^4\otimes \bbC^{4}$.}
		\label{fig:567=5x5-1}
	\end{figure}

	\begin{figure}[H]
		\centering
		\begin{tikzpicture}
		[
		box/.style={rectangle,draw=black,thick, minimum size=1cm},
		]
		
		\foreach \x in {0,1,...,4}{
			\foreach \y in {0,1,...,4}
			\node[box] at (\x,\y){};
		}
		\node[box]  at (0,3){1};
		\node[box]  at (1,3){1};
		\node[box]  at (1,2){2};
		\node[box]  at (2,2){2};
		\node[box]  at (2,1){3};
		\node[box]  at (3,1){3};
		\node[box]  at (3,0){4};
		\node[box]  at (0,0){4};
		\node[box]  at (0,2){5};
		\node[box]  at (0,1){5};
		\node[box]  at (1,1){6};
		\node[box]  at (1,0){6};
		\node[box]  at (2,3){7};
		\node[box]  at (2,0){7};
		\node[box]  at (3,3){8};
		\node[box]  at (3,2){8};
		\node[box] at (0,4){1};
		\node[box] at (1,4){1};
		\node[box] at (2,4){7};
		\node[box] at (3,4){8};
		\node[box] at (4,4){8};
		\node[box] at (4,3){8};
		\node[box] at (4,2){8};
		\node[box] at (4,1){3};
		\node[box] at (4,0){4};		
		\end{tikzpicture}
		\quad
		\begin{tikzpicture}
		[
		box/.style={rectangle,draw=black,thick, minimum size=1cm},
		]
		
		\foreach \x in {0,1,...,4}{
			\foreach \y in {0,1,...,4}
			\node[box] at (\x,\y){};
		}
		\node[box]  at (0,3){1};
		\node[box]  at (1,3){1};
		\node[box]  at (1,2){2};
		\node[box]  at (2,2){2};
		\node[box]  at (2,1){3};
		\node[box]  at (3,1){3};
		\node[box]  at (3,0){4};
		\node[box]  at (0,0){4};
		\node[box]  at (0,2){5};
		\node[box]  at (0,1){5};
		\node[box]  at (1,1){6};
		\node[box]  at (1,0){6};
		\node[box]  at (2,3){7};
		\node[box]  at (2,0){7};
		\node[box]  at (3,3){8};
		\node[box]  at (3,2){8};
		\node[box] at (0,4){1};
		\node[box] at (1,4){1};
		\node[box] at (2,4){7};
		\node[box] at (3,4){8};
		\node[box] at (4,4){9};
		\node[box] at (4,3){9};
		\node[box] at (4,2){9};
		\node[box] at (4,1){9};
		\node[box] at (4,0){4};		
		\end{tikzpicture}
		\quad
		\begin{tikzpicture}
		[
		box/.style={rectangle,draw=black,thick, minimum size=1cm},
		]
		
		\foreach \x in {0,1,...,4}{
			\foreach \y in {0,1,...,4}
			\node[box] at (\x,\y){};
		}
		\node[box]  at (0,3){1};
		\node[box]  at (1,3){1};
		\node[box]  at (1,2){2};
		\node[box]  at (2,2){2};
		\node[box]  at (2,1){3};
		\node[box]  at (3,1){3};
		\node[box]  at (3,0){4};
		\node[box]  at (0,0){4};
		\node[box]  at (0,2){5};
		\node[box]  at (0,1){5};
		\node[box]  at (1,1){6};
		\node[box]  at (1,0){6};
		\node[box]  at (2,3){7};
		\node[box]  at (2,0){7};
		\node[box]  at (3,3){8};
		\node[box]  at (3,2){8};
		\node[box] at (0,4){10};
		\node[box] at (1,4){10};
		\node[box] at (2,4){10};
		\node[box] at (3,4){10};
		\node[box] at (4,4){9};
		\node[box] at (4,3){9};
		\node[box] at (4,2){9};
		\node[box] at (4,1){9};
		\node[box] at (4,0){4};		
		\end{tikzpicture}
		\caption{U-tile structures with $8,9,10$-tiles in $\bbC^5\otimes \bbC^{5}$ based on the U-tile structure with $8$-tiles in $\bbC^4\otimes \bbC^{4}$.}
		\label{fig:8910=5x5}
	\end{figure}
   When $m\geq 6$, for each $5\leq t\leq 2(m-1)$, we  construct a U-tile structure with $t$-tiles in  $\bbC^m\otimes \bbC^m$ based on the U-tile structure $\cT$ with $t$-tiles in $\bbC^{m-1}\otimes\bbC^{m-1}$, by first appending a new row which is identical to the first row of length $m-1$ on the top, and a new column which is identical to the last column of length $m$ on the right. See Figure \ref{fig:567=5x5-1} for examples. For $t=2m-1$ and $2m$, we can construct U-tile structures with $t$-tiles in  $\bbC^m\otimes \bbC^m$ based on the U-tile structure with $2(m-1)$-tiles in $\bbC^{m-1}\otimes\bbC^{m-1}$. See Figures~\ref{Fig:meven} and~\ref{Fig:modd} for even $m$ and odd $m$, respectively. The new rows and new columns are on the top and right.
%

\begin{figure}[H]
	\centering
	\begin{tikzpicture}
	[
	box/.style={rectangle,draw=black,thick, minimum size=1cm},
	]
	
	\foreach \x in {0,1,...,5}{
		\foreach \y in {0,1,...,5}
		\node[box] at (\x,\y){};
	}
	\node[box] at (0,4){$i_1$};
	\node[box] at (1,4){$i_2$};
	\node[box] at (2,4){$\ldots$};
	\node[box] at (3,4){$i_{m-2}$};
	\node[box] at (4,4){$i_{m-1}$};
	\node[box] at (4,3){$j_{1}$};
	\node[box] at (4,2){$\vdots$};
	\node[box ] at (4,1){$j_{m-3}$};
	\node[box,] at (4,0){$j_{m-2}$};
	\node[box] at (0,5){5};
	\node[box] at (1,5){2$m$-1};
	\node[box] at (2,5){\ldots};
	\node[box] at (3,5){2$m$-1};
	\node[box] at (4,5){2$m$-1};
	\node[box] at (5,5){2$m$-1};
	\node[box] at (5,4){$i_{m-1}$};
	\node[box] at (5,3){$j_1$};
	\node[box] at (5,2){\vdots};
	\node[box] at (5,1){$j_{m-3}$};
	\node[box] at (5,0){4};		
	\end{tikzpicture}
	\quad
	\begin{tikzpicture}
	[
	box/.style={rectangle,draw=black,thick, minimum size=1cm},
	]
	
	\foreach \x in {0,1,...,5}{
		\foreach \y in {0,1,...,5}
		\node[box] at (\x,\y){};
	}
	\node[box] at (0,4){$i_1$};
	\node[box] at (1,4){$i_2$};
	\node[box] at (2,4){$\ldots$};
	\node[box] at (3,4){$i_{m-2}$};
	\node[box] at (4,4){$i_{m-1}$};
	\node[box] at (4,3){$j_{1}$};
	\node[box] at (4,2){$\vdots$};
	\node[box ] at (4,1){$j_{m-3}$};
	\node[box,] at (4,0){$j_{m-2}$};
	\node[box] at (0,5){5};
	\node[box] at (1,5){2$m$-1};
	\node[box] at (2,5){\ldots};
	\node[box] at (3,5){2$m$-1};
	\node[box] at (4,5){2$m$-1};
	\node[box] at (5,5){2$m$-1};
	\node[box] at (5,4){2$m$};
	\node[box] at (5,3){2$m$};
	\node[box] at (5,2){\vdots};
	\node[box] at (5,1){2$m$};
	\node[box] at (5,0){2$m$};		
	\end{tikzpicture}
	\caption{U-tile structures with $(2m-1),(2m)$-tiles in $\bbC^m\otimes \bbC^m$ based on the U-tile structure with $2(m-1)$-tiles in $\bbC^{m-1}\otimes \bbC^{m-1}$, where $m\geq 6$ is even.} \label{Fig:meven}
\end{figure}

\begin{figure}[H]
	\centering
	\begin{tikzpicture}
	[
	box/.style={rectangle,draw=black,thick, minimum size=1cm},
	]
	
	\foreach \x in {0,1,...,5}{
		\foreach \y in {0,1,...,5}
		\node[box] at (\x,\y){};
	}
	\node[box] at (0,4){$i_1$};
	\node[box] at (1,4){$i_2$};
	\node[box] at (2,4){$\ldots$};
	\node[box] at (3,4){$i_{m-2}$};
	\node[box] at (4,4){$i_{m-1}$};
	\node[box] at (4,3){$j_{1}$};
	\node[box] at (4,2){$\vdots$};
	\node[box ] at (4,1){$j_{m-3}$};
	\node[box,] at (4,0){$j_{m-2}$};
	\node[box] at (0,5){5};
	\node[box] at (1,5){$i_2$};
	\node[box] at (2,5){\ldots};
	\node[box] at (3,5){$i_{m-2}$};
	\node[box] at (4,5){$i_{m-1}$};
	\node[box] at (5,5){2$m$-1};
	\node[box] at (5,4){2$m$-1};
	\node[box] at (5,3){2$m$-1};
	\node[box] at (5,2){\vdots};
	\node[box] at (5,1){2$m$-1};
	\node[box] at (5,0){4};		
	\end{tikzpicture}
	\quad
	\begin{tikzpicture}
	[
	box/.style={rectangle,draw=black,thick, minimum size=1cm},
	]
	
	\foreach \x in {0,1,...,5}{
		\foreach \y in {0,1,...,5}
		\node[box] at (\x,\y){};
	}
	\node[box] at (0,4){$i_1$};
	\node[box] at (1,4){$i_2$};
	\node[box] at (2,4){$\ldots$};
	\node[box] at (3,4){$i_{m-2}$};
	\node[box] at (4,4){$i_{m-1}$};
	\node[box] at (4,3){$j_{1}$};
	\node[box] at (4,2){$\vdots$};
	\node[box ] at (4,1){$j_{m-3}$};
	\node[box,] at (4,0){$j_{m-2}$};
	\node[box] at (0,5){2$m$};
	\node[box] at (1,5){2$m$};
	\node[box] at (2,5){\ldots};
	\node[box] at (3,5){2$m$};
	\node[box] at (4,5){2$m$};
	\node[box] at (5,5){2$m$-1};
	\node[box] at (5,4){2$m$-1};
	\node[box] at (5,3){2$m$-1};
	\node[box] at (5,2){\vdots};
	\node[box] at (5,1){2$m$-1};
	\node[box] at (5,0){4};		
	\end{tikzpicture}
	\caption{U-tile structures with $(2m-1),(2m)$-tiles in $\bbC^m\otimes \bbC^m$ based on the U-tile structure with $2(m-1)$-tiles in $\bbC^{m-1}\otimes \bbC^{m-1}$, where $m\geq 6$ is odd.} \label{Fig:modd}
\end{figure}
	
	So far, we have constructed a UPB of size $(m^2-k)$ in $\bbC^m\otimes\bbC^m$ for $4\leq k\leq 2m-1$ and $m\geq 4$.  For the system $\bbC^m\otimes\bbC^n$ with $m\leq n$, a  U-tile structure  with $t$-tiles can be obtained from a U-tile structure $\cT$  with $t$-tiles in $\bbC^m\otimes\bbC^m$ by appending $n-m$ columns which are identical to the last column of $\cT$.
Hence by Theorem~\ref{Thm:U-tile}, there exists a UPB of size $(mn-k)$  in $\bbC^m\otimes\bbC^n$ for  $4\leq k\leq 2m-1$ and $4\leq m\leq n$.

By \cite{chen2013separability,feng2006unextendible}, there is no UPB of size $mn-1$, $mn-2$, $mn-3$ in $\bbC^m\otimes\bbC^n$ for $m,n\geq 3$.
To show that the maximum number of states in a UPB in $\bbC^m\otimes\bbC^m$,  $F(m,n)=mn-4$ for $3\leq m\leq n$, we only need
 a U-tile structure with $5$-tiles in $\bbC^m\otimes\bbC^n$ for $3\leq m\leq n$ by Theorem~\ref{Thm:U-tile}. See Figure~\ref{Fig:5tiles} for a construction.
\begin{figure}[H]
	\centering
	\begin{tikzpicture}
	[
	box/.style={rectangle,draw=black,thick, minimum size=1cm},
	]
	
	\foreach \x in {0,1,...,4}{
		\foreach \y in {0,1,...,4}
		\node[box] at (\x,\y){};
	}
	\node[box]  at (0,4){1};
	\node[box]  at (1,4){1};
	\node[box]  at (2,4){\dots};
	\node[box]  at (3,4){1};
	\node[box]  at (4,4){2};
	\node[box]  at (4,3){2};
	\node[box]  at (4,2){\vdots};
	\node[box]  at (4,1){2};
	\node[box]  at (4,0){3};
	\node[box]  at (3,0){3};
	\node[box]  at (2,0){\dots};
	\node[box]  at (1,0){3};
	\node[box]  at (0,0){4};
	\node[box]  at (0,1){4};
	\node[box]  at (0,2){\vdots};
	\node[box]  at (0,3){4};
	\node[box]  at (1,1){5};
	\node[box]  at (1,2){\vdots};
	\node[box]  at (1,3){5};
	\node[box]  at (2,1){\ldots};
	\node[box]  at (2,2){5};
	\node[box]  at (2,3){\ldots};
	\node[box]  at (3,1){5};
	\node[box]  at (3,2){\vdots};
	\node[box]  at (3,3){5};
	\draw (-1,4) node {0};
	\draw (-1,3) node {1};
	\draw (-1,2) node {\vdots};
	\draw (-1,1) node {$m-2$};
	\draw (-1,0) node {$m-1$};
	\draw (0,5) node {0};
	\draw (1,5) node {1};
	\draw (2,5) node {\ldots};
	\draw (3,5) node {$n-2$};
	\draw (4,5) node {$n-1$};
	\end{tikzpicture}
	\caption{A U-tile structure with 5-tiles in $\bbC^m\otimes \bbC^n$}\label{Fig:5tiles}
\end{figure}
\end{widetext}


\end{document}